\documentclass[11pt,a4paper]{article}
\usepackage[margin=25mm]{geometry}
\usepackage[T1]{fontenc}
\usepackage{lmodern}
\usepackage{microtype}
\usepackage{amsmath,amssymb,amsthm,mathtools}
\usepackage{booktabs,tabularx,array}
\usepackage{enumitem}
\usepackage{tikz}
\usetikzlibrary{arrows.meta,positioning,calc,shapes.geometric}
\usepackage{caption}
\usepackage{graphicx}
\usepackage{xcolor}
\PassOptionsToPackage{hyphens}{url}
\usepackage[colorlinks=true,linkcolor=blue!45!black,citecolor=blue!45!black,urlcolor=blue!45!black]{hyperref}
\hypersetup{pdftitle={A State-Space Model of Figured-Bass Realization},pdfsubject={Local musical constraints and polynomial-time realization},pdfauthor={Evan Unit Lim}}
\definecolor{voiceblue}{RGB}{34,80,125}
\definecolor{voicegreen}{RGB}{35,112,86}
\definecolor{voiceorange}{RGB}{169,88,22}
\definecolor{voicepurple}{RGB}{112,57,133}
\newtheorem{definition}{Definition}[section]
\newtheorem{proposition}[definition]{Proposition}
\newtheorem{theorem}[definition]{Theorem}
\newtheorem{corollary}[definition]{Corollary}

\newcommand{\ind}{\mathbf{1}}
\newcommand{\voices}{\mathcal{V}}
\newcommand{\states}{\mathcal{S}}
\newcommand{\notes}{\mathcal{P}}
\newcommand{\allowed}{\mathcal{M}}
\newcommand{\pc}{\operatorname{pc}}

\newcommand{\nt}[1]{\textup{#1}}
\newcommand{\figbass}[2]{\ensuremath{\begin{smallmatrix}#1\\#2\end{smallmatrix}}}
\newcommand{\voicing}[4]{(\nt{#1},\nt{#2},\nt{#3},\nt{#4})}
\newcolumntype{Y}{>{\raggedright\arraybackslash}X}
\setlist{itemsep=2pt,topsep=5pt}
\title{A State-Space Model of Figured-Bass Realization\\[3pt]
\large Local Constraints, Coupled Voices, and Polynomial-Time Solvability}
\author{Evan Unit Lim\\[3pt]\normalsize National Taiwan Normal University\\[2pt]\normalsize\href{mailto:elim@ntnu.edu.tw}{\texttt{elim@ntnu.edu.tw}}}
\date{}

\begin{document}
\maketitle
\vspace{-1.5em}
\begin{abstract}
Figured-bass realization can be described as a sequence of choices constrained
both within each sonority and between successive sonorities. This paper gives
an explicit mathematical model of a restricted, examination-style four-part
realization problem. Pitch spelling, range, chord membership, doubling,
omission, spacing, crossing, overlap, melodic motion, consecutive perfect
intervals, and selected resolution requirements are expressed as predicates.
We distinguish hard constraints from optional preference costs. Four labeled
notes are represented visually as the vertices of a quadrilateral and
computationally as one ordered voicing state. Legal progressions become paths
through a layered graph. We prove that feasibility and minimum-cost
realization are polynomial-time problems for a fixed number of voices with
explicit finite note domains and fixed local rules. For fixed ranges, a fixed
note alphabet, and adjacent-event rules, the number of graph operations is
linear in the number of events. Worked two-, four-, and eight-beat examples
illustrate legality, optimization, and the failure of a greedy choice. The
result concerns the stated formal model; it is not a claim that every musical
judgment is captured by local predicates.
\end{abstract}

\section{Scope and the realization question}

A figured-bass exercise supplies a bass line and interval indications, and
asks for compatible upper parts. The decisions are coupled: a note that is
acceptable in one voice can become unacceptable when another voice is chosen.
Nevertheless, coupling alone does not establish computational intractability.
The decisive questions are how many choices must be remembered at once and
how far a rule reaches along the score.

\begin{center}
\begin{minipage}{\linewidth}
\centering
\includegraphics[height=43mm]{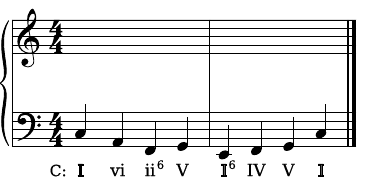}
\captionof{figure}{An original, unsolved exercise in C major, with space for the upper parts. The supplied bass and inversion figures are supplemented by Roman numerals for the reader; they do not specify a realization. Ordinary $5/3$ is left implicit. The tenor would be added above the bass on the lower staff, and soprano and alto on the blank upper staff. Section~\ref{sec:eight} supplies the boundary conditions and a worked realization of this bass.}
\label{fig:exercise}
\end{minipage}
\end{center}

Figured-bass realization is educationally useful because it brings chord
construction and the movement of individual voices into one practical task.
Karpinski argues that realizing basses and using figures to understand voice
leading should accompany the study of harmony \cite[Section 3]{karpinski}.
Reme\v{s} describes a historically informed approach in which figured bass,
chord spelling, voice leading, harmonic analysis, and keyboard and vocal
practice develop together as transferable skills \cite{remes}. The broader
integration of harmonic and linear thinking is also central to Aldwell,
Schachter, and Cadwallader's textbook \cite{aldwell}. These accounts motivate
studying the coordination required by the exercise; the present paper does
not test the educational effectiveness of a particular teaching method.

Figured-bass exercises also appear in music examinations worldwide, in
several distinct assessment traditions. ABRSM's Music Theory specification
includes four-part realization for SATB or keyboard in its Grade 6 content
\cite[pp.~19, 27]{abrsm}. College Board's AP Music Theory examination includes
part-writing from figured bass: its 2024 Question 5 rubric evaluates chord
spelling, spacing, doubling, and voice leading \cite{aptheory}. The Royal
Conservatory's ARCT Harmony \& Counterpoint syllabus includes keyboard-style
figured-bass realization \cite[p.~30]{rcm}. In Australia, AMEB's examiner
guidance explicitly addresses four-part vocal progressions over a supplied
figured bass \cite[p.~6]{ameb}. These examples document international use;
they do not imply that every music curriculum or examination applies the
same rules, texture, or marking scheme. For example, the AP question includes
a suspension, whereas our worked examples admit chord tones only. Its
point-based assessment also differs from the movement cost minimized here
\cite{aptheory}.

We study a deliberately stated SATB convention, rather than asserting one
universal set of examination rules. Upper-voice spacing and doubling appear
in both pedagogical treatments \cite[Sections 26.6--26.7]{hutchinson} and
examination rubrics \cite{aptheory,ameb}; the exact predicates below define
the problem analyzed here. Different syllabuses can replace these predicates without
changing the graph construction, provided their information requirements
remain within the stated assumptions. This is an expository formalization;
no claim of algorithmic novelty is made.

Initially, each event contains one sounding note in each voice, and all voices
advance together. In the worked examples an event is one beat. More generally,
the event grid should include every relevant onset, release, or change of
figure. A held note is represented by equality across successive events.
The number of events, not the notated number of bars, is the input-length
parameter. The examples have no rests, ornaments, or suspensions.

\begin{definition}[Decision and optimization versions]
Given a normalized bass-and-figures sequence of $n\ge1$ events, finite note
domains, prescribed endpoint restrictions, and a fixed local rule set, \emph{feasibility} asks whether at least
one realization satisfies every hard rule. \emph{Minimum-cost realization}
asks for a feasible realization minimizing a specified sum of local costs.
Its associated decision version asks whether such a realization has cost at
most a supplied threshold $K$.
\end{definition}

The class $\mathsf{P}$ formally concerns decision problems. We prove that
feasibility and the threshold version are in $\mathsf{P}$, and also give a
polynomial-time algorithm that constructs an optimal realization.

\subsection{Related work and contribution}

Computational harmonization has a substantial history. Ebcioglu's CHORAL
system represents musical knowledge through logical rules viewed at several
levels, including individual lines and the harmonic structure, and uses
backtracking with musical heuristics to harmonize a supplied melody
\cite{ebcioglu}. Its aim is stylistic chorale generation, with a broader
knowledge representation than the adjacent-event predicates studied here.

Pachet and Roy survey harmonization as a finite-domain constraint problem,
explicitly including figured bass and problems with a supplied voice
\cite{pachetroy}. They discuss local vertical and horizontal constraints,
finite-state approaches, and the use of chord variables whose domains are
constructed from permitted note combinations. That last construction is
especially close to the complete-voicing states below: grouping interacting
notes into one state is an established computational idea, not a new
harmonization algorithm introduced by this paper. The present account makes
the spelling, endpoint conditions, costs, and dependence on voice count and
memory length explicit, and connects these definitions to both staff notation
and the quadrilateral visualization.

Learned approaches address a different objective. DeepBach learns a model
from Bach chorales and generates music through pseudo-Gibbs sampling while
allowing user-imposed musical constraints \cite{deepbach}. Here the admissible
set and objective are specified before solving; the result is an exact
feasibility or minimum-cost guarantee for that formal set. It does not
establish stylistic equivalence to Bach, listener preference, or superiority
over an existing generation system.

The contribution is therefore an accessible mathematical formulation and a
parameter-explicit proof for a restricted exercise model. The tractability
argument rests on a fixed number of interacting voices and bounded temporal
dependence. Describing a musical task as a constraint problem alone does not
determine its complexity, and this proof does not classify every task treated
in the broader harmonization literature.

\section{Notes, figures, and complete voicings}
\label{sec:objects}

\subsection{Pitch must retain spelling}

Let $\notes$ be a finite set of spelled, octave-specific notes. A note is a
triple $p=(\ell,a,o)$ consisting of letter index $\ell\in\{0,\ldots,6\}$
for C through B, accidental $a\in\mathbb Z$, and octave $o$. The permitted
triples are explicitly bounded by the note domain. Set
\begin{align}
 m(p)&=12(o+1)+s_\ell+a,
 & (s_0,\ldots,s_6)&=(0,2,4,5,7,9,11),\label{eq:pitch}\\
 d(p)&=7o+\ell,
 & c(p)&=(\ell,a).
\end{align}
Thus $m(\nt{C4})=60$, $d$ records diatonic position, and $c$ records spelling
without octave. A bare letter such as $C$ in a chord-tone list abbreviates
the spelled class $(0,0)$, and similarly for the other letter names and
accidentals; it is not an integer semitone class. The sounding pitch class is $\pc(p)=m(p)\bmod12$.
Enharmonic notes may agree in $m$ while differing in $d$ or $c$; interval
quality and tonal function therefore cannot be recovered from $m$ alone.

Write the labeled voices in increasing register order as
$\voices=(B,T,A,S)$. A complete voicing is the ordered tuple
\begin{equation}
 X_t=(x_{B,t},x_{T,t},x_{A,t},x_{S,t})\in\notes^4.
 \label{eq:state}
\end{equation}
Write $m^{(4)}(X_t)=(m(x_{B,t}),m(x_{T,t}),m(x_{A,t}),m(x_{S,t}))$
for its componentwise semitone encoding. A voicing is not an unordered chord or a set of pitch classes. Register, voice
identity, and repeated notes all matter. Because the bass is supplied,
$x_{B,t}=b_t$; only three coordinates are normally selected.

\subsection{Normalizing the bass and figures}

At event $t$, let $H_t=(q_{1,t},\ldots,q_{\kappa_t,t})$ list the permitted distinct
spelled chord tones without octave. Let $\allowed_t$ be a set of permitted
multiplicity vectors in that order. We also record harmonic context and
explicit resolution flags where needed. These data form a normalized event
description $h_t$. In particular, $h_t$ supplies a spelled tonic class
$\tau_t$, any active leading-tone class $\lambda_t$, and the resolution flags
used below. The index $\kappa_t=|H_t|$ counts chord tones, not voices.

For an ordinary figure $f$, the letter relation above the bass obeys
\begin{equation}
 d(p)-d(b_t)\equiv f-1\pmod7.
 \label{eq:figure}
\end{equation}
The key signature and the figure's alterations determine the relevant
accidental. This congruence describes an interval class; by itself it does
not require a particular register or enforce the presence of the indicated
tone. Those requirements belong to the domain and multiplicity constraints.
Conventions that distinguish a literal ninth from a second need the
corresponding register or resolution condition as well.

For example, in C major, C3 with $\figbass{5}{3}$ gives $H_t=(C,E,G)$,
whereas E3 with $\figbass{6}{3}$ gives the same chord-tone collection but fixes
a different bass. Unfigured basses use the exercise's stated defaults.
Altered figures can introduce chromatic notes, so a fixed key does not imply
a seven-pitch-class note universe.

\begin{center}
\begin{minipage}{\linewidth}
\centering
\includegraphics[height=40mm]{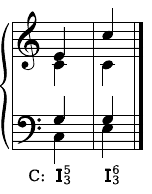}
\captionof{figure}{The normalization example in C major: C3 under I$^{5}_{3}$ and E3 under I$^{6}_{3}$ both give the spelled collection $H_t=(C,E,G)$. The expanded figures are shown here to expose the intervals above each bass. These are separate illustrative voicings, not an asserted legal voice-leading connection.}
\label{fig:normalization}
\end{minipage}
\end{center}

This paper analyzes realization \emph{after} this finite figure interpretation.
It does not treat unrestricted natural-language interpretation as a subproblem.
If several harmonic interpretations are allowed, their labels must be added
to the state whenever they imply different future obligations. Two identical
note tuples cannot then be merged solely on the basis of pitch.

\subsection{Shorthand conventions in figured-bass notation}
\label{sec:shorthand}

The tables in this paper deliberately display the expanded interval figures
so that readers can see the input to the mathematical model. Conventional
notation often abbreviates them. In the elementary triadic context assumed
here, an unmarked bass note normally implies $\figbass{5}{3}$; the two numbers
are usually omitted entirely. A first-inversion triad is commonly marked
with $6$ alone, with the third understood. This is an omission of written
symbols, not permission to omit the corresponding chord tone. Hutchinson
describes the common inversion symbols and their expanded forms
\cite[Sections 16.2 and 16.5]{hutchinson}.

\begin{table}[htbp]
\centering
\begin{tabular}{lll}
\toprule
Chord position & Expanded interval figures & Common shorthand\\
\midrule
Root-position triad & $\figbass{5}{3}$ & No figure\\[5pt]
First-inversion triad & $\figbass{6}{3}$ & $6$\\[5pt]
Second-inversion triad & $\figbass{6}{4}$ & $\figbass{6}{4}$\\[5pt]
Root-position seventh chord & $\begin{smallmatrix}7\\5\\3\end{smallmatrix}$ & $7$\\[5pt]
First-inversion seventh chord & $\begin{smallmatrix}6\\5\\3\end{smallmatrix}$ & $\figbass{6}{5}$\\[5pt]
Second-inversion seventh chord & $\begin{smallmatrix}6\\4\\3\end{smallmatrix}$ & $\figbass{4}{3}$\\[5pt]
Third-inversion seventh chord & $\begin{smallmatrix}6\\4\\2\end{smallmatrix}$ & $\figbass{4}{2}$ (also $2$)\\[5pt]
\bottomrule
\end{tabular}
\caption{Common pedagogical abbreviations. Written shorthand is expanded
before chord membership and multiplicity are tested. Historical notation
and explicit alterations may require additional context.}
\label{tab:shorthand}
\end{table}

An accidental written without a numeral commonly modifies the implied third;
an accidental beside a numeral modifies that interval. Such alterations,
and indications for suspensions or changing intervals over a held bass,
must be retained during normalization. An arbitrary unfigured historical
bass does not uniquely determine its harmony: the default used here belongs
to the stated exercise convention. In the score figures, Roman numerals
identify the interpreted harmony and attached figures identify its inversion:
for example, ii$^6$ means a first-inversion supertonic triad. Root-position
triads normally appear as I, IV, or V, with $5/3$ implicit; $6/3$ is abbreviated
to $6$. Figure~\ref{fig:normalization} deliberately shows the expanded form,
as do the input tables. The prefix ``C:'' identifies C major. These analytical
labels are supplied as reading aids, not as claims that Roman numerals form
part of every historical figured-bass source. The letter-and-octave tables and graphs
remain a parallel, notation-independent presentation of the same notes.

\begin{table}[htbp]
\centering\small
\begin{tabularx}{\linewidth}{lY}
\toprule
Notation & Meaning\\
\midrule
$p,\notes;\ m,d,c$ & Spelled note and finite note domain; semitone position,
diatonic position, and spelled class.\\
$\voices=(B,T,A,S),\ X_t$ & Ordered voices and complete voicing at event $t$.\\
$b_t,h_t,H_t,\allowed_t$ & Supplied bass, normalized context, chord-tone list,
and permitted multiplicities.\\
$\tau_t,\lambda_t$ & Local tonic and active leading-tone spelled classes.\\
$A_t,\states_t$ & Within-event hard penalty and vertically legal states.\\
$B_t,S_t,W_t$ & Transition hard penalty, preference cost, and joint weight.\\
$\mathcal I,\mathcal F$ & Permitted initial and final states.\\
$J,D_t$ & Total cost and minimum prefix cost.\\
$n,M,K$ & Number of events, upper-voice domain-size bound, and cost threshold.\\
\bottomrule
\end{tabularx}
\caption{Notation used in the model. Hard penalties take values $0$ or
$+\infty$; preference costs are finite. The precise constraints and recurrence
are defined in the following sections.}
\label{tab:notation}
\end{table}

\section{Constraints within a voicing}
\label{sec:vertical}

\subsection{Range, crossing, and spacing}

The example convention uses the following inclusive ranges:
\begin{center}
\begin{tabular}{lcc}
\toprule
Voice & Note range & Semitone bounds $[L_v,U_v]$\\
\midrule
Bass & E2--E4 & $[40,64]$\\
Tenor & C3--G4 & $[48,67]$\\
Alto & G3--D5 & $[55,74]$\\
Soprano & C4--G5 & $[60,79]$\\
\bottomrule
\end{tabular}
\end{center}

\begin{center}
\begin{minipage}{\linewidth}
\centering
\includegraphics[width=.96\linewidth]{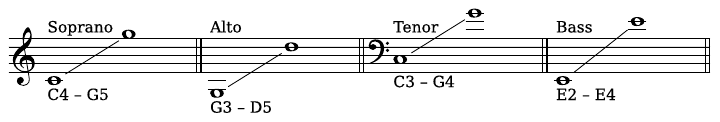}
\captionof{figure}{The inclusive voice ranges used in this paper on one
continuous staff, with double barlines separating the voices. A straight
line joins each lowest and highest note to indicate the intervening range,
not a glissando. All pitches sound as written; the clef changes from treble
to bass for tenor and bass. These are range endpoints, not chords or a
progression.}
\label{fig:ranges}
\end{minipage}
\end{center}
For every voice, require
\begin{equation}
 L_v\le m(x_{v,t})\le U_v.
\end{equation}
No crossing means
\begin{equation}
 m(x_{B,t})\le m(x_{T,t})\le m(x_{A,t})\le m(x_{S,t}).
\end{equation}
Equality is permitted: two distinct voices may occupy the same note.
Spacing is constrained by
\begin{equation}
 m(x_{A,t})-m(x_{T,t})\le12,
 \qquad m(x_{S,t})-m(x_{A,t})\le12.
\end{equation}
There is no additional octave limit between bass and tenor.

\subsection{Membership, doubling, and omission}

In the chord-tone-only model, require $c(x_{v,t})\in H_t$ for every voice.
For each spelled tone $q$, define
\begin{equation}
 N_t(q;X_t)=\sum_{v\in\voices}\ind\{c(x_{v,t})=q\}.
\end{equation}
The joint doubling-and-omission condition is
\begin{equation}
 \bigl(N_t(q_{1,t};X_t),\ldots,N_t(q_{\kappa_t,t};X_t)\bigr)
 \in\allowed_t.
 \label{eq:multiplicity}
\end{equation}
Counts include the bass. They distinguish two occurrences of a tone in
different octaves from two different chord tones.

\begin{table}[htbp]
\centering
\begin{tabularx}{\linewidth}{Yl}
\toprule
Example permission, in root--third--fifth(--seventh) order & Vector\\
\midrule
Complete triad, root doubled & $(2,1,1)$\\
Triad, fifth omitted and root tripled & $(3,1,0)$\\
Complete seventh chord & $(1,1,1,1)$\\
Seventh chord, fifth omitted and root doubled & $(2,1,0,1)$\\
Complete triad, third doubled & $(1,2,1)$\\
\bottomrule
\end{tabularx}
\caption{Examples of multiplicity patterns, not unrestricted musical
permissions. An event admits only the vectors explicitly included in
$\allowed_t$.}
\label{tab:multiplicity}
\end{table}

The worked root-position triads use only $(2,1,1)$. Their first-inversion
exceptions are specified in Section~\ref{sec:eight}. No omissions are enabled in the
worked examples. If $\lambda_t$ is an active leading tone, impose
\begin{equation}
 N_t(\lambda_t;X_t)\le1.
\end{equation}
For the examples, $\lambda_t=B$ in C major. A tonicization-aware model may
instead supply the locally active leading tone in $h_t$.

Let $A_t(X)$ be zero when the prescribed bass, domain, and all enabled
within-event conditions hold, and $+\infty$ otherwise. Define the legal
voicing set by
\begin{equation}
 \states_t=\{X:A_t(X)=0\}.
 \label{eq:legalstates}
\end{equation}

\section{Constraints between voicings}
\label{sec:horizontal}

Fix adjacent states $X=X_t$ and $Y=X_{t+1}$. Write
$\Delta_v=m(y_v)-m(x_v)$ and $\delta_v=d(y_v)-d(x_v)$.
All rules in this section inspect these states together with the supplied
contexts $h_t,h_{t+1}$.

\subsection{Melodic movement and overlap}

The deliberately restrictive melodic convention for the examples permits
repetition, minor and major seconds and thirds, perfect fourths, and perfect
fifths. Define
\begin{equation}
 \mathcal L=\{(0,0),(1,1),(1,2),(2,3),(2,4),(3,5),(4,7)\}.
\end{equation}
For each voice, including the supplied bass, require
\begin{equation}
 (\delta_v,\Delta_v)\in\mathcal L\cup(-\mathcal L).
 \label{eq:melody}
\end{equation}
The signed pairs prevent an augmented or diminished interval from being
mistaken for an allowed interval merely because it has a small semitone span.
An exercise allowing larger melodic leaps would enlarge this relation.

For each neighboring voice pair $u<v$, prohibit overlap by requiring
\begin{equation}
 m(y_u)\le m(x_v),\qquad m(y_v)\ge m(x_u).
 \label{eq:overlap}
\end{equation}
Crossing compares voices within an event; overlap compares their registers
across two events. Neither condition implies the other.

\subsection{Consecutive and direct perfect intervals}

For a lower note $p$ and an upper note $q$, define
\begin{align}
 P_5(p,q)&\iff \exists \rho\in\mathbb Z_{\ge0}:\quad
 \begin{cases}d(q)-d(p)=7\rho+4,\\m(q)-m(p)=12\rho+7,\end{cases}\\
 P_8(p,q)&\iff \exists \rho\in\mathbb Z_{\ge0}:\quad
 \begin{cases}d(q)-d(p)=7\rho,\\m(q)-m(p)=12\rho.\end{cases}
\end{align}
The same octave count must satisfy both equalities in each predicate; this
retains exact spelled interval quality even for unusual accidental values.
These families include compound fifths and octave-equivalent intervals;
$P_8$ includes unisons. Similar motion means $\Delta_u\Delta_v>0$, contrary
motion means $\Delta_u\Delta_v<0$, and both voices move exactly when the
product is nonzero. For every one of the $\binom42=6$ voice pairs, forbid
\begin{equation}
 P_j(x_u,x_v)\land P_j(y_u,y_v)\land(\Delta_u\Delta_v\ne0),
 \qquad j\in\{5,8\}.
 \label{eq:parallel}
\end{equation}
Thus the model prohibits consecutive perfect intervals of the same family
in both similar and contrary motion. Stationary repetitions and oblique
motion are excluded from this prohibition. A fifth followed by an octave
belongs to different families and is not forbidden by this predicate alone.
We reserve ``parallel'' for the similar-motion cases and use ``consecutive''
for the broader rule.

This convention follows the explicit treatment of contrary-motion fifths
and octaves in the AP Music Theory 2024 Question 5 rubric
\cite[Section III.D.1 and DCVLE no.~3]{aptheory}.
AMEB also identifies consecutive fifths and octaves as errors, but the cited
guidance does not separately specify the contrary-motion case \cite{ameb}.
These sources support the stated convention without establishing how nearly
universal it is across examination systems.

\begin{center}
\begin{minipage}{\linewidth}
\centering
\includegraphics[width=.85\linewidth]{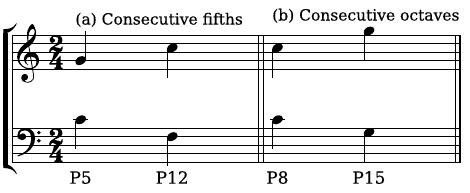}
\captionof{figure}{Consecutive perfect intervals in contrary motion, forbidden
by Equation~\eqref{eq:parallel}. In (a), C4--G4 becomes F3--C5: a perfect
fifth becomes a perfect twelfth (P12), both in the $P_5$ family. In (b),
C4--C5 becomes G3--G5: an octave becomes a double octave (P15), both in the
$P_8$ family. Each lower voice descends while the upper voice ascends.
These isolated voice pairs illustrate interval predicates, not complete
harmonies; interval names therefore replace Roman-numeral analysis.}
\label{fig:contrary-perfect}
\end{minipage}
\end{center}

For direct outer-voice fifths and octaves, forbid
\begin{equation}
 (\Delta_B\Delta_S>0)\land
 \bigl(P_5(y_B,y_S)\lor P_8(y_B,y_S)\bigr)
 \land (|\delta_S|>1).
 \label{eq:direct}
\end{equation}
Thus similar motion into a perfect interval is permitted when the soprano
moves by step. This is a stated convention, not an assertion that every
pedagogical system treats direct intervals identically.

\subsection{Conditional resolution requirements}

Resolution requirements must have explicit triggers. If the context requires
voice $v$'s leading tone $\lambda_t$ to resolve to its tonic $\tau_t$, require
\begin{equation}
 c(y_v)=\tau_t,\qquad
 (\delta_v,\Delta_v)=(1,1).
 \label{eq:leading}
\end{equation}
In the examples this is enabled for every B in a V chord followed by I or
I$^6$ in C major. Frustrated leading-tone resolutions are not allowed in this
particular convention.

When a chordal seventh is required to resolve downward by step, require
\begin{equation}
 (\delta_v,\Delta_v)\in\{(-1,-1),(-1,-2)\}.
\end{equation}
This predicate is available to the model, but the triadic examples do not
activate it. Dissonance treatment and special-chord resolutions can likewise
be expressed as context-conditioned relations. Merely identifying a tone as
a seventh does not establish that immediate downward resolution is mandatory
in every possible context.

Let $B_t(X,Y)$ be zero when all enabled transition conditions hold, and
$+\infty$ otherwise. No musical rule is silently inferred beyond the
conditions specified in the instance and the fixed rule set.

\section{Hard constraints and preference costs}

Hard rules answer whether a realization is admissible. Soft preferences rank
admissible realizations. For illustration, use total upper-voice movement:
\begin{equation}
 S_t(X,Y)=\sum_{v\in\{T,A,S\}}|m(y_v)-m(x_v)|.
 \label{eq:soft}
\end{equation}
The bass is omitted because its motion is supplied and contributes the same
amount to every realization of an instance. This cost does not measure all
aspects of musical quality and is not an examination marking scheme.

Local numerical costs are supplied as nonnegative integers or rationals with
finite binary encodings. The illustrative movement cost is integer-valued.
Let $\mathcal I\subseteq\states_1$ and $\mathcal F\subseteq\states_n$
be the permitted initial and terminal sets. They are given explicitly or by
polynomial-time membership tests, such as fixing an opening voicing or a
final soprano note. An unrestricted boundary uses the entire corresponding
legal-state set. These restrictions do not change the vertically legal sets
$\states_t$. The objective is
\begin{equation}
 J(X_1,\ldots,X_n)=
 \sum_{t=1}^n A_t(X_t)+
 \sum_{t=1}^{n-1}\bigl(B_t(X_t,X_{t+1})+S_t(X_t,X_{t+1})\bigr).
 \label{eq:objective}
\end{equation}
A realization is feasible exactly when $X_1\in\mathcal I$,
$X_n\in\mathcal F$, and $J(X_1,\ldots,X_n)<+\infty$. Optimization minimizes
$J$ over sequences satisfying both endpoint restrictions. For pure feasibility,
set $S_t=0$. Other local preferences may be substituted, including finite
penalties for selected stylistic features. Replacing a hard prohibition with
a finite penalty changes the problem: a sufficiently costly violation may
then be selected. Removing forbidden vertices or edges is exactly equivalent
to the $+\infty$ convention and avoids numerical approximations to infinity.

\section{Why must the graph account for complete voicings?}
\label{sec:coupling}

Start with a graph whose vertices are individual notes and whose edges are
candidate melodic moves. Four labeled travelers follow this graph in
synchrony. Range and individual melodic restrictions can be tested on one
traveler. Doubling, spacing, overlap, and consecutives depend on several
travelers together. Consequently, four independently optimized paths with
fixed individual-edge weights do not, in general, solve the realization
problem.

The parallel-fifths predicate alone displays this dependence. Suppose two
voices initially sing C4 and G4. Allow the lower voice to move to D4 or E4,
and the upper voice to A4 or B4. All four individual moves are permitted
melodic intervals. Testing only parallel fifths gives
\begin{center}
\begin{tabular}{lcc}
\toprule
 & Upper: G4$\to$A4 & Upper: G4$\to$B4\\
\midrule
Lower: C4$\to$D4 & Forbidden & Allowed\\
Lower: C4$\to$E4 & Allowed & Forbidden\\
\bottomrule
\end{tabular}
\end{center}
This isolates one rule; the four destination combinations are not claimed
to be realizations of one prescribed triad. It shows why that rule has a
joint dependence before harmonic filtering.

\begin{center}
\begin{minipage}{\linewidth}
\centering
\includegraphics[width=.96\linewidth]{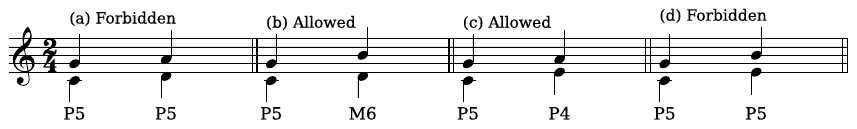}
\captionof{figure}{The four combinations in the parallel-fifths table,
read row by row. Each case starts with C4--G4. In (a), D4--A4 retains a
perfect fifth (P5), as does E4--B4 in (d); both voices move upward, producing
forbidden parallel fifths. In (b), D4--B4 is a major sixth (M6); in (c),
E4--A4 is a perfect fourth (P4). These two cases pass the parallel-fifths
test. Double barlines separate independent cases. Interval labels are used
here because the dyads do not prescribe complete harmonies; no other rule
is being tested.}
\label{fig:parallel-fifths}
\end{minipage}
\end{center}

An additive assignment of finite-or-infinite weights to individual moves
cannot reproduce this table. The two allowed off-diagonal combinations
require all four individual moves to have finite weights; both diagonal
combinations would then also have finite weights. The missing information is
which other move occurs simultaneously.

There are two exact ways to retain it:
\begin{enumerate}
\item Keep the note graph, but evaluate joint constraints on all four
      travelers and their simultaneous moves.
\item Form a graph of complete voicing configurations, making each joint
      transition an ordinary graph edge.
\end{enumerate}
These are representations of the same coupled choices. The second makes
the information needed by a shortest-path algorithm explicit. It is not
necessary to materialize every configuration in memory; they may be generated
when needed. The tuple is sufficient for the adjacent-event model, not a
claim that no more compact equivalent representation exists.

\subsection{The quadrilateral and the dot}
\label{sec:polygon}

Visualize a four-part voicing as a quadrilateral with one labeled corner for
each voice. Each corner carries its octave-specific note. At a coarser scale,
treat the whole quadrilateral as one dot labeled by the ordered tuple.
The dot retains the complete internal configuration; it does not average
the pitches or discard their assignments. Drawing a tuple as one dot does
not reduce the number of possible tuples: distinct assignments remain
distinct candidate nodes.

\begin{center}
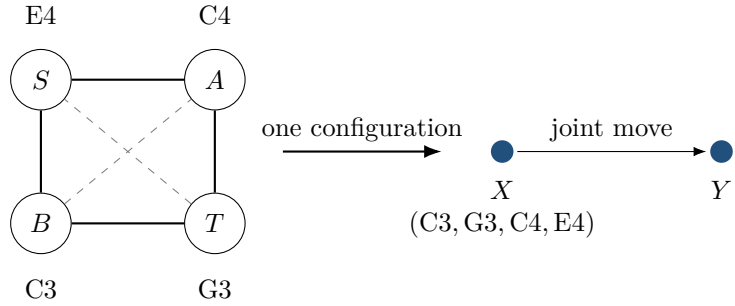

\begin{minipage}{\linewidth}
\centering
\begin{tikzpicture}[>=Latex,voice/.style={circle,draw,fill=white,minimum size=8mm,font=\small},
    every node/.style={font=\small}]
\node[voice] (b) at (0,0) {$B$};
\node[voice] (t) at (2.3,0) {$T$};
\node[voice] (a) at (2.3,1.9) {$A$};
\node[voice] (s) at (0,1.9) {$S$};
\draw[thick] (b)--(t) (t)--(a) (a)--(s) (s)--(b);
\draw[dashed,gray] (b)--(a) (t)--(s);
\node[below=2mm of b] {C3};
\node[below=2mm of t] {G3};
\node[above=2mm of a] {C4};
\node[above=2mm of s] {E4};
\draw[->,thick] (3.2,.95)--node[above,align=center]{one configuration}(5.3,.95);
\node[circle,fill=voiceblue,inner sep=3pt] (x) at (6.1,.95) {};
\node[below,align=center] at (6.1,.7) {$X$\\$(\nt{C3},\nt{G3},\nt{C4},\nt{E4})$};
\draw[->] (6.3,.95)--node[above] {joint move}(8.8,.95);
\node[circle,fill=voiceblue,inner sep=3pt] at (9,.95) {};
\node[below] at (9,.7) {$Y$};
\end{tikzpicture}
\captionof{figure}{The quadrilateral is a visual container for four labeled voice
positions. Its sides and diagonals indicate the six possible voice-pair
interactions. A node packages the full tuple; its placement on the page
does not encode pitch or discard internal information.}
\label{fig:polygon}
\end{minipage}
\end{center}

Both diagonals matter: a plain four-edge cycle would omit two voice pairs.
Even the complete pairwise picture does not express everything, since a
multiplicity rule counts all four voices at once. Such a rule belongs to
the whole configuration. If two voices sing in unison, their labeled corners
remain distinct.

Five voices can be drawn as a pentagon, and $k$ voices as a labeled $k$-gon,
with all relevant pairwise links and any whole-configuration predicates.
This is a representational metaphor, not a claim that harmony has Euclidean
polygon geometry. Mathematically the object is a labeled product state.

\section{The layered configuration graph}
\label{sec:graph}

Create one layer for each event, with vertex set
\begin{equation}
 V_t=\{(t,X):X\in\states_t\}.
\end{equation}
The time label matters: the same voicing at different events is a different
vertex. Connect $(t,X)$ to $(t+1,Y)$ exactly when $B_t(X,Y)=0$, and give the
edge weight $S_t(X,Y)$. Equivalently, start with every possible adjacent-layer
edge and assign forbidden edges weight $+\infty$.

Attach a source to each $X\in\mathcal I$ and each $X\in\mathcal F$
to a sink with zero-weight edges. A specified opening voicing or final
soprano note is therefore just a boundary restriction. With no prescribed
boundary, all legal states in the relevant layer are admitted.

\begin{proposition}[Exact correspondence]
For the adjacent-event rule model, feasible realizations correspond
bijectively to finite source-to-sink paths in the layered configuration graph.
The path weight equals the realization's finite preference cost.
\end{proposition}
\begin{proof}
A feasible realization chooses one state in each $\states_t$. Its adjacent
pairs pass every transition test, so they supply a path. Conversely, a path
selects a legal state at every event and only permitted transitions between
events. These are all the hard constraints in the model, including the
boundary conditions; thus the resulting realization is feasible. The time
labels and ordered tuples make both constructions inverse to each other.
Finally, both costs sum precisely the same local quantities $S_t$.
\end{proof}

The fixed-note-graph picture is still available: keep a
single universe $\notes$ with candidate moves $\notes\times\notes$, including
self-loops for repeated pitches. At event $t$ the four travelers occupy $X_t$,
and their simultaneous move is assigned the joint weight
$W_t(X,Y)=B_t(X,Y)+S_t(X,Y)$ for $X\in\states_t$ and
$Y\in\states_{t+1}$. If arbitrary note tuples are retained instead, use
$W_t(X,Y)=A_{t+1}(Y)+B_t(X,Y)+S_t(X,Y)$, initialize with $A_1(X_1)$,
and impose the same endpoint restrictions. The layered configuration graph
is the time expansion of these joint states. Weights may depend on time and state, but computing a
weight may not hide an unspecified search through an entire history.

\section{Dynamic programming and polynomial-time solvability}
\label{sec:complexity}

For $X\in\states_t$, let $D_t(X)$ be the smallest cost of a prefix satisfying
the initial restriction and every hard rule through event $t$, ending at $X$.
The terminal restriction is applied only at the last layer. Initialize
$D_1(X)=0$ for $X\in\mathcal I$ and $D_1(X)=+\infty$ for the other
$X\in\states_1$. For each $Y\in\states_{t+1}$, set
\begin{equation}
 D_{t+1}(Y)=\min_{X\in\states_t}
 \left[D_t(X)+B_t(X,Y)+S_t(X,Y)\right].
 \label{eq:dp}
\end{equation}
The minimum over an empty set is $+\infty$. The optimum is
$\min_{X\in\mathcal F}D_n(X)$.
If that value is infinite, the instance is infeasible. Remembering a
minimizing predecessor reconstructs a realization.

\begin{proposition}[Correctness of the recurrence]
Equation~\eqref{eq:dp} computes the minimum prefix cost at every layer.
\end{proposition}
\begin{proof}
The initialization describes every permitted one-event prefix. Assume the
claim at event $t$. Every feasible prefix ending at $Y$ at event $t+1$ has a
penultimate state $X$ with a permitted transition to $Y$. By the induction
hypothesis its earlier cost is at least $D_t(X)$. Conversely, any finite
$D_t(X)$ is realized by a feasible prefix that can be extended whenever the
transition is legal. Minimizing over $X$ therefore gives exactly the optimum
for $Y$.
\end{proof}

The argument uses a specific information property: with adjacent-event
rules and supplied context, two prefixes ending in the same complete state
have the same permitted continuations. Only the cheaper prefix need be
retained for optimization. This is a shortest-path application of Bellman's dynamic-programming
principle, as developed for routing problems in his 1958 article
\cite{bellman}. Our acyclic, event-layered graph permits a single forward
pass; the recurrence is not a greedy choice of the cheapest next chord.

\begin{theorem}[Polynomial-time realization]
\label{thm:poly}
Suppose there are four labeled voices, the bass and normalized contexts
$h_t$ are supplied and fixed, each upper-voice domain has at most $M$ explicitly listed notes, and all constraints involve
one event or two adjacent events. Suppose also that evaluating the fixed
predicates, endpoint membership tests, and local costs takes polynomial
time in their encoded inputs.
Then feasibility, threshold-cost feasibility, and construction of a
minimum-cost realization have polynomial-time algorithms. With constant-time
local evaluations, the number of candidate-state and transition operations
is $O(nM^6)$.
\end{theorem}
\begin{proof}
The three upper voices produce at most $M^3$ candidate tuples per event.
Filtering them costs at most $nM^3$ local state evaluations. Each pair of
successive layers has at most $M^3M^3=M^6$ candidate transitions. Generating
these and applying Equation~\eqref{eq:dp} takes at most $(n-1)M^6$ local transition
evaluations and relaxations. Hence the stated operation bound follows. If
each evaluation has polynomial cost, multiplying by that cost preserves
polynomial time. The final minimum answers feasibility and, by comparison
with $K$, threshold feasibility. Predecessor pointers reconstruct an
optimal realization in another $O(n)$ steps.
\end{proof}

Two adjacent distance layers require $O(M^3)$ stored values. Storing
predecessors for reconstruction requires $O(nM^3)$ entries. The graph need
not be stored explicitly. If the bass were also chosen, the analogous
elementary bound would be $O(nM^8)$, still polynomial for four voices.

If harmonic interpretation is also a choice, suppose at most $H$ explicitly
listed context labels are available per event. Retaining the chosen label
with the note tuple gives at most $HM^3$ states and $O(nH^2M^6)$ candidate
connections, provided compatibility and transition tests remain polynomial.
This counts labeled realizations: two label sequences can describe the same
sequence of pitches. The theorem's $M^3$ state bound assumes one supplied
context per event. Succinctly encoded or unrestricted interpretation systems
need a separate complexity analysis.

\begin{corollary}[Fixed ranges and fixed local rules]
\label{cor:fixed}
With a fixed finite spelled-note alphabet, fixed vocal ranges, and fixed
adjacent-event predicates, feasibility requires $O(n)$ graph operations.
Minimum-cost realization also requires $O(n)$ graph operations for fixed
local cost evaluation.
\end{corollary}

There are two encoding qualifications. First, $M$ must be an explicit domain
size or a fixed constant. A range of exponentially many pitches described by
two binary endpoints cannot silently be treated as an explicitly listed
domain. Second, arithmetic has a bit cost. If nonnegative integer local
costs have at most $L$ bits, path sums need $O(L+\log n)$ bits. For rational
costs with numerators and positive denominators of at most $L$ bits, a path
uses at most $n-1$ such costs. The product of their denominators is a common
denominator of $O(nL)$ bits, and the corresponding summed numerator has
$O(nL+\log n)$ bits. Exact addition and comparison therefore remain
polynomial-time operations. The graph operation bound remains as above, and
the total bit complexity is polynomial, including comparison with the
binary-encoded threshold $K$.
The linear claim is a graph-operation claim, not an unqualified assertion
of linear bit time for arbitrarily large numeric weights.

Even when the number of complete realizations grows exponentially with $n$,
the algorithm does not enumerate them: it merges prefixes that have the same
sufficient ending state. The fixed number of interacting voices is central.
For $k$ voices with one supplied bass, the corresponding bound is
$O(nM^{2(k-1)})$ local operations. It is polynomial in $n$ and $M$ for each
fixed $k$, but exponential in $k$ when $M>1$. That observation does not itself
prove hardness when $k$ is variable.

\paragraph{Finite hearing range does not bound the number of parts.}
A finite physical ensemble has a finite number of performers and notated
parts. In a complexity analysis, however, variable $k$ means arbitrarily
large \emph{finite} instances, not a literally infinite orchestra. Hearing
range alone does not supply a universal constant upper bound on $k$.
First, a bounded interval of frequencies is not itself a finite pitch
alphabet; a finite tuning grid and spelling convention must also be chosen.
Second, even if the resulting alphabet has $D$ notes, multiple labeled parts
may occupy the same note. Adding more instruments in unison does not require
any new pitch. Thus the finite-pitch assumption can bound the domain size
without bounding the number of travelers in the graph.

If one additionally required every part to occupy a distinct pitch at every
event, the pigeonhole principle would give $k\le D$. Orchestral doubling and
the permitted unisons in this paper do not satisfy that extra restriction.
Nor does inability to hear every simultaneous line separately remove those
lines from the notated score. A fixed ensemble size $k\le k_{\max}$ can legitimately
be imposed as an application-specific assumption, but it must be stated
independently of hearing range. This preserves the fixed-$k$ conclusion of
Corollary~\ref{cor:fixed}, without extending it to unrestricted orchestral
part counts merely on acoustic grounds.

\section{Worked examples}
\label{sec:examples}

All examples use C major, the ranges in Section~\ref{sec:vertical}, the exact
melodic relation in Equation~\eqref{eq:melody}, all six consecutive-perfect-interval tests, the overlap and
direct-interval tests, and strict leading-tone resolution on V--I and
V--I$^6$. Consecutive perfect intervals in contrary motion are prohibited
as specified in Equation~\eqref{eq:parallel}; no additional melody-recovery
rule is enabled. A beat is one simultaneous event. The opening voicing is fixed to
\begin{equation}
 X_1=\voicing{C3}{G3}{C4}{E4},\qquad
 m^{(4)}(X_1)=(48,55,60,64).
\end{equation}
\begin{center}
\begin{minipage}{\linewidth}
\centering
\includegraphics[height=34mm]{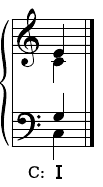}
\captionof{figure}{The fixed opening voicing $X_1$ as a short score: bass C3, tenor G3, alto C4, soprano E4. In C major this is root-position I, with the root doubled. Upper stems belong to soprano and tenor; lower stems belong to alto and bass.}
\label{fig:opening}
\end{minipage}
\end{center}

Each example is a separate instance: octave choices in a supplied bass are
part of its input and need not agree across examples. Tables list voices
from soprano down to bass, whereas tuples always run bass to soprano.
Companion score figures use a conventional two-staff SATB short score:
soprano and alto on the treble staff, tenor and bass on the bass staff, with
upper-voice stems up and lower-voice stems down on each staff. Quarter notes
represent the events; the two-beat examples are engraved in $2/4$ and the
four- and eight-beat examples in $4/4$. These metrical groupings are for
presentation and add no metrical constraints to the model.

\subsection{Two beats: legal chords can have an illegal connection}

\begin{center}
\begin{minipage}{\linewidth}
\centering
\includegraphics[height=18mm]{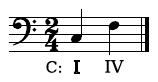}
\captionof{figure}{Two-beat input: bass C3--F3 with harmony I--IV. Both root-position triads imply $5/3$. The three upper voices are still to be selected.}
\label{fig:input-two}
\end{minipage}
\end{center}

Take bass C3--F3 with figures $\figbass{5}{3}$ at both events. Compare the
second voicings $Y=\voicing{F3}{A3}{C4}{F4}$ and
$Z=\voicing{F3}{C4}{F4}{A4}$.
\begin{table}[htbp]
\centering
\begin{tabular}{lccc}
\toprule
 & Beat 1: $X_1$ & Beat 2: $Y$ & Beat 2: $Z$\\
\midrule
Soprano & E4 & F4 & A4\\
Alto & C4 & C4 & F4\\
Tenor & G3 & A3 & C4\\
Bass & C3 & F3 & F3\\
Harmony & I & IV & IV\\
\bottomrule
\end{tabular}
\caption{Both destination chords pass the within-event tests, but only one
of these two displayed connections is legal.}
\end{table}

For $X_1\to Y$, the movements in bass-to-soprano order are $(5,2,0,1)$.
All hard rules pass and the preference cost is $2+0+1=3$. The outer voices
arrive at an octave by similar motion, which is permitted here because the
soprano moves by step.

For $X_1\to Z$, the movements are $(5,5,5,5)$. The bass and tenor move from
C3--G3 to F3--C4, giving parallel perfect fifths. The bass and alto move from
C3--C4 to F3--F4, giving parallel octaves. There is also alto--soprano overlap:
the new alto F4 lies above the previous soprano E4. Thus this connection has
infinite hard-constraint cost despite its individually legal destination.

\begin{center}
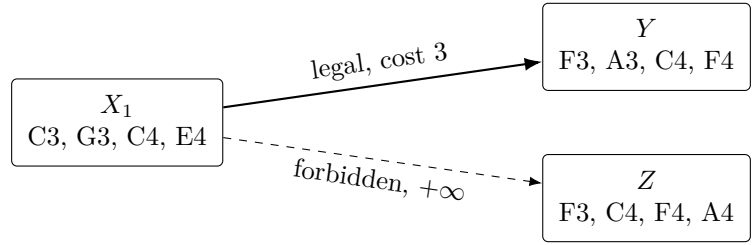

\begin{minipage}{\linewidth}
\centering
\begin{tikzpicture}[>=Latex,box/.style={draw,rounded corners=2pt,align=center,font=\small,inner sep=6pt}]
\node[box] (x) at (0,0) {$X_1$\\C3, G3, C4, E4};
\node[box] (y) at (7,1) {$Y$\\F3, A3, C4, F4};
\node[box] (z) at (7,-1) {$Z$\\F3, C4, F4, A4};
\draw[->,thick] (x)--node[above,sloped,font=\small] {legal, cost 3}(y);
\draw[->,dashed] (x)--node[below,sloped,font=\small] {forbidden, $+\infty$}(z);
\end{tikzpicture}
\captionof{figure}{A selected fragment of the two-layer graph. The dashed edge denotes
a rejected candidate and is absent from the feasible graph.}
\end{minipage}
\end{center}

For completeness, exhaustive checking gives exactly three legal second
voicings from this opening:
\begin{equation}
\begin{array}{c|c}
\text{Second voicing in }(B,T,A,S)\text{ order}&\text{Cost}\\\hline
\voicing{F3}{F3}{A3}{C4}&9\\
\voicing{F3}{F3}{C4}{A4}&7\\
\voicing{F3}{A3}{C4}{F4}&3
\end{array}
\end{equation}
The first two use permitted bass--tenor unisons. This list certifies that
the displayed legal choice $Y$ is optimal for the complete two-beat instance,
not merely preferable to the rejected choice $Z$.

\begin{center}
\begin{minipage}{\linewidth}
\centering
\begin{minipage}[t]{.48\linewidth}
\centering
\textbf{(a) Admissible: $X_1\to Y$}\\[4pt]
\includegraphics[height=40mm]{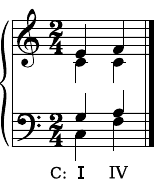}
\end{minipage}\hfill
\begin{minipage}[t]{.48\linewidth}
\centering
\textbf{(b) Rejected: $X_1\to Z$}\\[4pt]
\includegraphics[height=40mm]{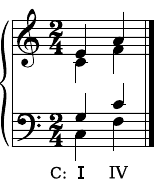}
\end{minipage}
\captionof{figure}{Staff-notation companions to the two-beat table and graph. Both
examples use the supplied bass C3--F3 and imply $5/3$ at both beats. The
right-hand candidate contains the parallel fifths, parallel octaves, and
overlap described in the text; it is shown for comparison, not as an
admissible realization.}
\label{fig:score-two}
\end{minipage}
\end{center}

\subsection{Four beats: a complete local progression}

\begin{center}
\begin{minipage}{\linewidth}
\centering
\includegraphics[height=21mm]{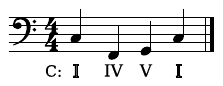}
\captionof{figure}{Four-beat input: the supplied bass C3--F2--G2--C3 and its harmony I--IV--V--I. The root-position figures $5/3$ are implicit.}
\label{fig:input-four}
\end{minipage}
\end{center}

Take the supplied bass C3--F2--G2--C3, with $\figbass{5}{3}$ throughout,
giving I--IV--V--I. One optimal realization is:
\begin{table}[htbp]
\centering
\begin{tabular}{lcccc}
\toprule
Beat & 1 & 2 & 3 & 4\\
\midrule
Soprano & E4 & F4 & D4 & E4\\
Alto & C4 & C4 & B3 & C4\\
Tenor & G3 & A3 & G3 & G3\\
Bass & C3 & F2 & G2 & C3\\
Figures & $\figbass{5}{3}$ & $\figbass{5}{3}$ & $\figbass{5}{3}$ & $\figbass{5}{3}$\\
Harmony & I & IV & V & I\\
Incoming cost & -- & 3 & 6 & 3\\
\bottomrule
\end{tabular}
\caption{A four-beat minimum-cost realization. Its total upper-voice
movement is $3+6+3=12$.}
\label{tab:four}
\end{table}

The final alto B3--C4 satisfies the enabled leading-tone resolution.
All four chords have their roots doubled. There are exactly two feasible
complete realizations from the fixed opening under the stated rules; the
minimum cost is 12. The final soprano is not constrained to tonic, so this
example is not being presented as a perfect authentic cadence.

The state counts make the computation visible:
\begin{center}
\begin{tabular}{lrrrr}
\toprule
Beat & 1 & 2 & 3 & 4\\
\midrule
Vertically legal states $|\states_t|$ & 10 & 8 & 8 & 10\\
Reachable states from the fixed opening & 1 & 2 & 2 & 2\\
Minimum finite prefix cost & 0 & 3 & 9 & 12\\
\bottomrule
\end{tabular}
\end{center}
The algorithm retains a separate value for each reachable state. The row of
minimum prefix costs is only a summary; retaining just those four scalar
minima would lose the information needed to continue the calculation.

\begin{center}
\begin{minipage}{\linewidth}
\centering
\includegraphics[height=44mm]{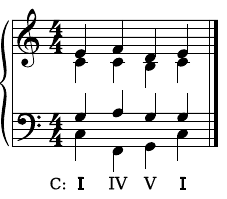}
\captionof{figure}{Staff notation for the four-beat realization in
Table~\ref{tab:four}. The notes and voice assignments are identical to the
table. All four chords use the unmarked $5/3$ convention.}
\label{fig:score-four}
\end{minipage}
\end{center}

\subsection{Eight beats: inversions, a boundary condition, and a greedy trap}
\label{sec:eight}

\begin{center}
\begin{minipage}{\linewidth}
\centering
\includegraphics[height=22mm]{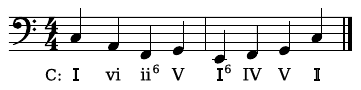}
\captionof{figure}{Eight-beat input: C3--A2--F2--G2--E2--F2--G2--C3, with harmonic analysis and inversion figures. The first-inversion chords at beats 3 and 5 are ii$^6$ and I$^6$; the other triads are in root position.}
\label{fig:input-eight}
\end{minipage}
\end{center}

Consider
\begin{equation}
 \mathrm{I}\;\text{--}\;\mathrm{vi}\;\text{--}\;\mathrm{ii}^{6}
 \;\text{--}\;\mathrm{V}\;\text{--}\;\mathrm{I}^{6}
 \;\text{--}\;\mathrm{IV}\;\text{--}\;\mathrm{V}\;\text{--}\;\mathrm{I},
\end{equation}
with bass C3--A2--F2--G2--E2--F2--G2--C3. The final soprano is required to
be C5. At beat 3, ii$^6$ uses the third-doubled vector $(1,2,1)$ in
$(D,F,A)$ order. At beat 5, I$^6$ uses $(2,1,1)$ in $(C,E,G)$ order.
These explicit inversion choices belong to this example; no general
first-inversion doubling rule is assumed.

\begin{table}[htbp]
\centering
\setlength{\tabcolsep}{5pt}
\begin{tabular}{lcccccccc}
\toprule
Beat & 1 & 2 & 3 & 4 & 5 & 6 & 7 & 8\\
\midrule
Soprano & E4 & A4 & A4 & B4 & C5 & C5 & B4 & C5\\
Alto & C4 & E4 & F4 & D4 & C4 & F4 & D4 & E4\\
Tenor & G3 & C4 & D4 & G3 & G3 & A3 & G3 & G3\\
Bass & C3 & A2 & F2 & G2 & E2 & F2 & G2 & C3\\
Figures & $\figbass{5}{3}$ & $\figbass{5}{3}$ & $\figbass{6}{3}$ & $\figbass{5}{3}$ & $\figbass{6}{3}$ & $\figbass{5}{3}$ & $\figbass{5}{3}$ & $\figbass{5}{3}$\\
Harmony & I & vi & ii$^6$ & V & I$^6$ & IV & V & I\\
Incoming cost & -- & 14 & 3 & 12 & 3 & 7 & 6 & 3\\
\bottomrule
\end{tabular}
\caption{An eight-beat optimal realization with fixed opening and final
soprano C5. The total cost is 48. All notes are natural.}
\label{tab:eight}
\end{table}

Every adjacent upper-voice gap remains within an octave. The soprano B4--C5 resolves
the leading tone at both V--I$^6$ and the final V--I. The bass also satisfies
the specified melodic relation, including the opening descending minor
third and the final ascending perfect fourth. The tenor's D4--G3 is a
permitted descending perfect fifth. No recovery-after-a-leap rule has been
added; enabling such a rule would require rechecking the example.

\begin{center}
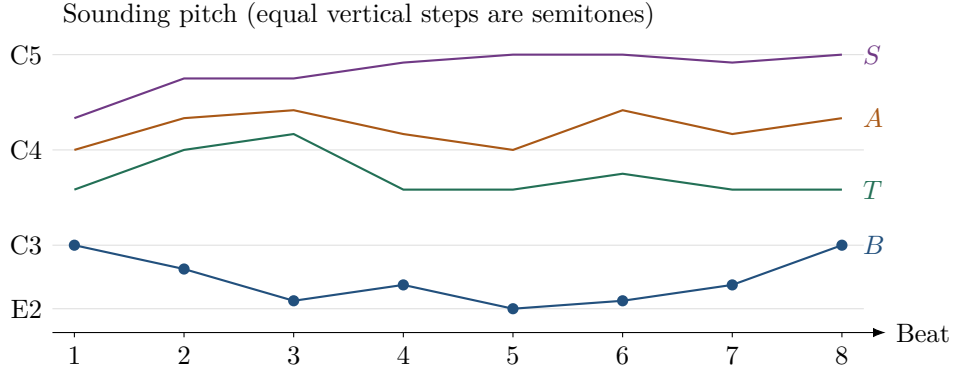

\begin{minipage}{\linewidth}
\centering
\begin{tikzpicture}[x=1.45cm,y=.105cm,>=Latex,font=\small]
\foreach \p/\lab in {40/E2,48/C3,60/C4,72/C5}{
  \draw[gray!25] (.8,\p)--(8.2,\p);
  \node[left] at (.8,\p) {\lab};
}
\draw[->] (.8,37)--(8.4,37) node[right]{Beat};
\foreach \t in {1,...,8}{\draw (\t,37)--(\t,36.4);\node[below] at (\t,36.4){\t};}
\draw[voiceblue,thick] plot[mark=*,mark size=1.7pt] coordinates {(1,48)(2,45)(3,41)(4,43)(5,40)(6,41)(7,43)(8,48)};
\draw[voicegreen,thick] plot[mark=square*,mark size=1.7pt] coordinates {(1,55)(2,60)(3,62)(4,55)(5,55)(6,57)(7,55)(8,55)};
\draw[voiceorange,thick] plot[mark=triangle*,mark size=2pt] coordinates {(1,60)(2,64)(3,65)(4,62)(5,60)(6,65)(7,62)(8,64)};
\draw[voicepurple,thick] plot[mark=diamond*,mark size=2pt] coordinates {(1,64)(2,69)(3,69)(4,71)(5,72)(6,72)(7,71)(8,72)};
\node[right,text=voiceblue] at (8.1,48) {$B$};
\node[right,text=voicegreen] at (8.1,55) {$T$};
\node[right,text=voiceorange] at (8.1,64) {$A$};
\node[right,text=voicepurple] at (8.1,72) {$S$};
\node[anchor=west] at (.8,77) {Sounding pitch (equal vertical steps are semitones)};
\end{tikzpicture}
\captionof{figure}{Four synchronized voice trajectories for Table~\ref{tab:eight}. Each
vertical slice specifies one complete configuration node. Connecting lines
indicate voice motion, not continuous glissandi.}
\label{fig:eight}
\end{minipage}
\end{center}

This instance also distinguishes dynamic programming from a locally greedy
method. At beat 2 exactly two voicings are reachable from the fixed opening:
\begin{center}
\begin{tabular}{lcc}
\toprule
Beat-2 voicing in $(B,T,A,S)$ order & Immediate cost & Completion to final C5\\
\midrule
\voicing{A2}{A3}{C4}{E4} & 2 & None\\
\voicing{A2}{C4}{E4}{A4} & 14 & Feasible, minimum total 48\\
\bottomrule
\end{tabular}
\end{center}
Under the stated consecutive-interval rule, the cheaper first move has
no full-length continuation, even without the final C5 restriction. Choosing
only that move therefore loses every complete solution. Dynamic programming retains both
current states until their possible continuations are resolved.

\begin{center}
\setlength{\tabcolsep}{5pt}
\begin{tabular}{lrrrrrrrr}
\toprule
Beat & 1 & 2 & 3 & 4 & 5 & 6 & 7 & 8\\
\midrule
Vertically legal states & 10 & 8 & 8 & 8 & 8 & 8 & 8 & 10\\
Forward-reachable states & 1 & 2 & 5 & 4 & 5 & 7 & 5 & 5\\
Minimum finite prefix cost & 0 & 2 & 5 & 13 & 16 & 21 & 39 & 42\\
\bottomrule
\end{tabular}
\end{center}
These forward counts are taken \emph{before} imposing the terminal soprano
filter. At beat 8 only two reachable states have soprano C5; their minimum
cost is 48. The unrestricted minimum 42 ends at a different soprano pitch
and is not a solution of the specified instance. There are 14 feasible
complete paths satisfying both boundary conditions.

\begin{center}
\begin{minipage}{\linewidth}
\centering
\includegraphics[height=48mm]{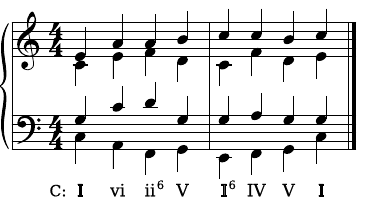}
\captionof{figure}{Staff notation for the eight-beat realization in
Table~\ref{tab:eight}, grouped into two bars of $4/4$. The harmony row labels
every event, with ii$^6$ and I$^6$ at beats 3 and 5; the remaining Roman
numerals imply $5/3$. The soprano ends on the required C5.}
\label{fig:score-eight}
\end{minipage}
\end{center}

\subsection{Verification and interpretation of the examples}

The reported state counts and optima were checked by direct enumeration of
the finite natural-note domains using the predicates printed in this paper.
A separate complete-path traversal, retaining all legal prefixes instead of
merging them, confirmed the feasible-path counts $3$, $2$, and $14$ and the
minimum costs $3$, $12$, and $48$, respectively. This second check verifies
the dynamic-programming optima against exhaustive search for these small
instances; it does not turn their musical conventions into universal rules.
No music-analysis library or external realization implementation was used
to derive the rules, select the notes, or verify feasibility and optimality.
The verified note lists were subsequently encoded in MEI and engraved with
Verovio solely to produce the companion staff-notation figures. Roman numerals
and inversion figures were attached explicitly from the stated harmony,
rather than inferred by a music-analysis library. The supplied bass figures,
full realizations, and range endpoints were checked against their note lists.
The engraving stage does not alter the musical model or its computed results.
A reproduction procedure appears in Section~\ref{sec:reproduce}.

The examples establish concrete behavior of the model. Their finite sizes
do not establish an asymptotic complexity class; Theorem~\ref{thm:poly} supplies that
proof for arbitrary $n$ under the stated assumptions.

\section{Rules extending beyond adjacent events}
\label{sec:extensions}

A state containing only the current voicing is not sufficient for every
part-writing instruction.
For example, a prepared suspension requires information about preparation,
dissonance, and resolution. Recovery after a leap compares two successive
melodic movements. Such rules can still have bounded memory.

As an illustrative recovery convention, if a voice leaps by more than a
third from event $t-1$ to $t$, require a contrary diatonic step next:
\begin{equation}
 |\delta_{v,t-1}|>2\ \Longrightarrow\
 |\delta_{v,t}|=1\ \land\ \Delta_{v,t-1}\Delta_{v,t}<0.
\end{equation}
Here $\delta_{v,t}=d(x_{v,t+1})-d(x_{v,t})$ and
$\Delta_{v,t}=m(x_{v,t+1})-m(x_{v,t})$. Admissible step quality is still checked by the melodic relation. This is an optional, explicitly
chosen rule, not part of the preceding examples. Suspension modeling also
requires admitting the designated non-chord tone during its dissonant event;
adding only a resolution test while retaining strict chord membership would
be inconsistent.

\begin{proposition}[A fixed number of remembered events]
If every constraint and cost inspects at most $r$ consecutive events, with
fixed $r\ge2$, then the configuration-path method still gives a
polynomial-time algorithm for fixed voice count and explicit note domains.
With a fixed note universe it uses $O(n)$ graph operations.
\end{proposition}
\begin{proof}
Use a state consisting of the previous $r-1$ complete voicings. If a single
event has at most $N$ states, there are at most $N^{r-1}$ such history states.
Appending one new voicing gives at most $N$ successors, and every newly
completed rule window can now be evaluated. Thus at most $O(nN^r)$ candidate
extensions are needed, apart from the cost of evaluating the rules. Since
$r$ and the voice count are fixed and $N\le M^3$ for SATB with supplied bass,
the bound is polynomial. Fixed domains make $N$ constant. Initialization and
terminal checks handle obligations at the boundaries.
\end{proof}

More generally, a rule can be represented by a finite memory state, such as
an unresolved-tendency flag. That memory becomes another coordinate of the
configuration. The argument remains polynomial when the additional state
space can be generated and processed in polynomial time. A fixed collection
of finite-state monitors preserves a fixed-width layered model.

An unrestricted requirement relating arbitrarily many distant events cannot
automatically be handled by the one-voicing recurrence. Some global
requirements remain tractable after adding a counter or another compact
summary; others may not. The computational effect must be proved for the
particular rule. Similarly, arbitrary input-dependent rule programs could
hide hard computations inside one transition test and fall outside the
fixed-rule assumptions.

\section{Discussion}

The quadrilateral picture separates two useful levels of description: notes
belong to labeled voices, and a simultaneous assignment of all voices is a
configuration. A transition between configurations is evaluated jointly.
Drawing a configuration as a dot packages this information without making
the voices independent.

Under fixed voice count, finite explicit domains, and local evaluable rules,
the formal realization problem is in $\mathsf{P}$. Fixed physical ranges and
a fixed finite spelling vocabulary give the stronger fixed-width result.
This classification applies to the precisely stated decision problem, and
does not assert that all aspects of good composition, stylistic judgment,
or every examination rubric are local. It also does not equate polynomial
time with small practical constants.

The model provides a basis for further work: additional musical rules can be
written as predicates, assigned the memory they actually require, and tested
against examples. Whenever that memory remains bounded, the same layered
construction applies. Any broader complexity claim should follow from the
resulting formal problem rather than from the mere number of possible scores.

\section*{Acknowledgements}
This research began in 2019 while the author was a mathematics tutor at
Berklee College of Music under the supervision of Dr.\ Jennifer Beauregard.
Parts of the research were conducted while the author was a student at the
College of William and Mary and at the Georgia Institute of Technology,
and while serving as a research associate in Dr.\ Cathy S.\ J.\ Fann's
laboratory at Academia Sinica.

The author acknowledges the use of GPT-6 Astra (OpenAI) in the preparation
of this manuscript. Specifically, the AI tool assisted with the drafting,
revision, and editing of the text, as well as with the literature review.

The mathematical content and proof structures presented in this paper were
initially derived by the author. GPT-6 Astra was used to identify edge cases,
refine mathematical notation, and suggest improvements to the proofs.
All AI-assisted suggestions and revisions were subsequently verified by the
author, who takes full responsibility for the final content.

\appendix
\section{Reproducibility specification}
\label{sec:reproduce}

For the examples, form each voice's domain by retaining only natural notes
whose semitone numbers lie within its inclusive range. At event $t$, restrict
that domain to the chord-tone spellings listed in Table~\ref{tab:inputs}, fix the
bass, enumerate the Cartesian product, and apply every within-event test.
Use the melodic relation, overlap rule, six consecutive-perfect-interval tests, direct outer
interval test, and triggered leading-tone resolutions for each adjacent pair.

\begin{table}[htbp]
\centering
\begin{tabular}{llll}
\toprule
Event type & Ordered chord tones & Multiplicity & Figures\\
\midrule
I & $(C,E,G)$ & $(2,1,1)$ & $\figbass{5}{3}$\\
IV & $(F,A,C)$ & $(2,1,1)$ & $\figbass{5}{3}$\\
V & $(G,B,D)$ & $(2,1,1)$ & $\figbass{5}{3}$\\
vi & $(A,C,E)$ & $(2,1,1)$ & $\figbass{5}{3}$\\
ii$^6$ & $(D,F,A)$ & $(1,2,1)$ & $\figbass{6}{3}$\\
I$^6$ & $(C,E,G)$ & $(2,1,1)$ & $\figbass{6}{3}$\\
\bottomrule
\end{tabular}
\caption{Normalized event templates used in the examples. Each row has a
single permitted multiplicity vector. The bass pitches are supplied in the
corresponding example tables.}
\label{tab:inputs}
\end{table}

\begin{samepage}
\noindent\textbf{Dynamic programming.}
\begin{enumerate}
\item Generate the legal voicing sets $\states_1,\ldots,\states_n$.
\item Set $D_1(X_1)=0$ for the prescribed opening, and set every other
      value to $+\infty$.
\item For each successive layer and each candidate destination $Y$, examine
      all currently reachable predecessors $X$. If the transition is legal,
      compare $D_t(X)+S_t(X,Y)$ with the best recorded value for $Y$ and retain
      a minimizing predecessor.
\item Minimize over the permitted terminal states. Follow predecessor links
      backward to recover one optimal path.
\end{enumerate}
\end{samepage}
\pagebreak[0]

\begin{samepage}\noindent\textbf{Exhaustive comparison for the small examples.}
\begin{enumerate}
\item Start with the one-element path containing the fixed opening.
\item Extend every current path by every legal next voicing. Retain distinct
      paths even when they end at the same state.
\item At the last event, apply the terminal condition, count the surviving
      paths, and calculate the minimum of their summed movement costs.
\end{enumerate}
The terminal condition is unrestricted for two and four beats and is
$x_{S,8}=\nt{C5}$ for eight beats. Counts of vertically legal states are taken
before initial or terminal filtering. Counts of reachable states apply the
fixed opening but, unless expressly stated otherwise, precede terminal
filtering. These conventions are necessary to reproduce the tables.\par
\end{samepage}

\end{document}